\documentclass[letterpaper, 10 pt, conference]{ieeeconf}  % 
\IEEEoverridecommandlockouts                              % 
\usepackage{graphicx,subcaption}
\usepackage{amsmath,amssymb}
\usepackage{mathrsfs}
\usepackage{booktabs}
\usepackage{color}
\usepackage{mathrsfs}
\usepackage{algorithm}

\newtheorem{theorem}{Theorem}

\newtheorem{proposition}{Proposition}
\newtheorem{lemma}{Lemma}
\newtheorem{definition}{Definition}
\newtheorem{example}{Example}
\newtheorem{remark}{Remark}

\newcommand{\TwoOne}[2]
{\begin{bmatrix}
{#1} \\
{#2}
\end{bmatrix}
}
\newcommand{\OneTwo}[2]
{\begin{bmatrix} {#1} & {#2}
\end{bmatrix}
}
\newcommand{\TwoTwo}[4]
{\begin{bmatrix}
{#1} & {#2} \\
{#3} & {#4}
\end{bmatrix}
}

\newcommand{\Ltwo}{\boldsymbol{\rm L}_{2}}
 
\newcommand{\Ltwoe}{\boldsymbol{\rm L}_{2e}} 
\title{\LARGE \bf
Passivity Compensation: A Distributed Approach for Consensus Analysis in Heterogeneous Networks
}

\author{Yongkang Su, Sei Zhen Khong and Lanlan Su% <-this % stops a space
\thanks{Y. Su is with the School of Electrical and Electronic Engineering, University of Sheffield, Sheffield, UK
        {\tt\small ysu34@sheffield.ac.uk}}%
\thanks{S. Z. Khong is with the Department of Electrical Engineering, National Sun Yat-sen University, Kaohsiung 80424, Taiwan
        {\tt\small szkhong@mail.nsysu.edu.tw}}%
\thanks{L. Su is with the Department of Electrical and Electronic Engineering, University of Manchester, Manchester, UK
        {\tt\small lanlan.su@manchester.ac.uk}}%
}

\begin{document}

\maketitle
\thispagestyle{empty}
\pagestyle{empty}

%%%%%%%%%%%%%%%%%%%%%%%%%%%%%%%%%%%%%%%%%%%%%%%%%%%%%%%%%%%%%%%%%%%%%%%%%%%%%%%%
\begin{abstract}
This paper investigates a passivity-based approach to output consensus analysis in heterogeneous networks composed of non-identical agents coupled via nonlinear interactions, in the presence of measurement and/or communication noise. Focusing on agents that are input-feedforward passive (IFP), we first examine whether a shortage of passivity in some agents can be compensated by a passivity surplus in others, in the sense of preserving the passivity of the transformed open-loop system defined by the agent dynamics and network topology. We show that such compensation is only feasible when at most one agent lacks passivity, and we characterise how this deficit can be offset using the excess passivity within the group of agents. For general networks, we then investigate passivity compensation within the feedback interconnection by leveraging the passivity surplus in the coupling links to locally compensate for the lack of passivity in the adjacent agents. In particular, a distributed condition, expressed in terms of passivity indices and coupling gains, is derived to ensure output consensus of the interconnected network.
\end{abstract}


%%%%%%%%%%%%%%%%%%%%%%%%%%%%%%%%%%%%%%%%%%%%%%%%%%%%%%%%%%%%%%%%%%%%%%%%%%%%%%%%
\section{Introduction}
In the subject of feedback stability analysis, the energy-based notion of passivity and the associated passivity theorem play a crucial role \cite{willems1972dissipative1}. The theorem states that if two open-loop systems are passive, and one possesses an excess of passivity, then the stability of their feedback interconnection can be established. This classical result has been extended by introducing quantitative measures of passivity, such as the input-feedforward passivity (IFP) index and the output-feedback passivity (OFP) index \cite{DesVid75}. These indices provide a more flexible framework for feedback stability analysis by enabling passivity “compensation”: when one system exhibits a shortage of passivity, it can be compensated by the passivity surplus of the other system \cite{van2000l2}.

In the literature on network consensus, a widely adopted approach is to transform the problem into a feedback stability analysis via a projection operation \cite{arcak2007passivity,scardovi2010synchronization,hamadeh2011global}. This transformation typically involves decomposing the network dynamics into a feedback interconnection of two open-loop systems: one representing the collection of agent dynamics and the other reflecting the graph topology and edge dynamics. The seminal work \cite{chopra2006passivity} establishes that if each agent’s dynamics are passive and they are diffusively coupled through a connected graph, the network will achieve consensus. Following the above-mentioned work, consensus problem of network with maximal equilibrium independent passive agents is addressed in \cite{burger2014duality}. This result was further extended in \cite{jain2018regularization,sharf2020geometric}, where the passivity-based analysis is generalised to scenarios in which all agents may lack passivity. In such cases, a suitable controller augments the agent to achieve maximal equilibrium independent passivity.  These studies address the consensus problem by showing that the open-loop system reflecting the graph topology is (strictly) passive, while the other open-loop system is rendered passive by either assuming all agents possess the relevant passivity property or through controller design. In contrast, our work focusses on analysing consensus in networks of agents that may lack passivity, without relying on local passivating controllers.

% Following the above-mentioned work, a number of significant research results have been developed; see, for example, \cite{arcak2007passivity,burger2014duality}.

% This result was further extended in \cite{li2019consensus}, where the passivity-based analysis is generalised to scenarios in which all agents may lack passivity. In such cases, the surplus of passivity in the coupling graph is used to compensate for this shortage, thereby enabling consensus.

In this work, we consider heterogeneous networks, i.e., networks composed of agents with different dynamics or characteristics, in which agents are coupled through sector-bounded interactions. Motivated by simple examples of two-agent and three-agent networks, we explore whether a shortage of passivity in some agents within a group can be meaningfully compensated by the passivity surplus of other agents, in a manner that facilitates consensus analysis.  A key observation from passivity theory is that the passivity property of a dynamical system is preserved under symmetric transformations of its input and output variables. Particularly, through pre- and post-multiplication by the graph incidence matrix and its transpose, this principle can be applied to network dynamics, thereby preserving passivity \cite{bai2011cooperative}. However, it remains an open question whether the passivity of the open-loop system — defined by the collective agent dynamics in conjunction with an input-output transformation determined by the network topology — can be ensured through passivity compensation between agents. If this were achievable, it would relax classical passivity conditions, demonstrating that even when some agents lack passivity, the transformed open-loop system could remain passive without requiring any controller design. In that case, the passivity theorem could be invoked to establish consensus immediately, provided that the edge dynamics are strictly passive.
We show that such compensation within the collection of agent dynamics alone is possible only in a highly restricted case: at most one agent may lacks passivity, and its shortage can be compensated by the excess passivity of other agents. 

Beyond this, for general networks, we turn our attention to the passivity surplus present in the coupling links. Key contributions along this research line include \cite{scardovi2010synchronization,li2019consensus,sharf2019network}. However, the consensus conditions proposed in these works are centralised, as they rely on global knowledge of agents and network topology, which limits their scalability and practicality in large-scale or distributed settings. In this work, we demonstrate that the passivity surplus in coupling links can be used to compensate the local lack of passivity in the agents connected by those links. Specifically, we derive a distributed condition, expressed in terms of passivity indices and coupling gains, under which consensus is ensured.

The remainder of the paper is organised as follows. Section~\ref{sec: Preliminaries} introduces the notation and preliminary concepts. Motivating examples  are presented in Section~\ref{sec: motivating examples}. Section~\ref{sec: problem formulation} formally states the problem. The main results are developed in Section~\ref{sec: main results}. In particular, Subsection~\ref{sec: Limitations of Generalisation to Arbitrary Networks} addresses passivity compensation among agents, while Subsection~\ref{sec: Passivity Compensation in Feedback Connection} focuses on passivity compensation between agents and their coupling links. Concluding remarks are provided in Section~\ref{sec: Conclusion}.
 %%%%%%%%%%%%%%%%%%%%%%%%%%%%%%%%%%%%%%%%%%%%%%%%
 
\section{Preliminaries}\label{sec: Preliminaries}

\subsection{Notation}
Let $\mathbb{R}$ be the set of real numbers. For a matrix $A$, denote by $A^{\top}$ and $\text{rank}(A)$ its transpose and its rank, respectively. Let $\textbf{1}_m := [1, \ldots ,1]^{\top} \in {\mathbb{R}^m}$. Given scalars ${{a_1}, \ldots ,{a_m}}$, let the column vector ${\rm col}\left( {{a_1}, \ldots ,{a_m}} \right) := {\left[ {a_1, \ldots ,a_m} \right]^{\top}}$ and $\text{diag}\{a_1,\dots, a_m\}$ the diagonal matrix with its $i$th diagonal entry being $a_i$. Given a symmetric matrix $A=A^{\top}$, we use $A\succ 0$ (resp., $A\succcurlyeq 0$) to denote that $A$ is positive definite (resp., positive semi-definite). 
Define the signal space $\Ltwo =\{ {x:\left[ {0,\infty } \right) \to {\mathbb{R}^m}|{{\left\| x \right\|}^2}: = \int_0^\infty  {{{\left| {x(t)} \right|}^2}dt < \infty } } \}$ where $|\cdot|$ denotes the Euclidean norm.
For any $x:\left[ {0,\infty } \right) \to {\mathbb{R}^m}$, denote by $P_T$ the truncation operator so that $\left( {{P_T}x} \right)(t) = x(t)$ for $t \le T$ and $\left( {{P_T}x} \right)(t) = 0$ for $t>T$. Define $\Ltwoe$ as $\Ltwoe =\{ {x:\left[ {0,\infty } \right) \to {\mathbb{R}^m}|{P_T}x \in \Ltwo ,\forall T \ge 0}\}$. Given $x \in \Ltwoe$ and $ T \ge 0$, ${\left\| x \right\|_T} := {( {\int_0^T {{{\left| {x(t)} \right|}^2}dt} } )^{\frac{1}{2}}}$. Given $x,y \in \Ltwoe $ and $T \ge 0$, ${\left\langle {x,y} \right\rangle _T} := \int_0^T {{x^{\top}}(t)y(t)dt}$. An operator $H:\Ltwoe \to \Ltwoe$ is said to be causal if ${P_T}H{P_T} = {P_T}H,\,\forall T \in \mathbb{R}$. 
% All operators $H:\Ltwoe \to \Ltwoe$ considered in this work are assumed to map $0$ to $0$. \textcolor{red}{Are you sure this is respected in your examples?}

\subsection{Graph Theory}\label{subsetion_Graph_Theory}
A graph is defined by $\mathcal{G} = (\mathcal{N},\mathcal{E} )$, where $\mathcal{N} = \{ 1, \ldots ,n\} $ is the set of nodes and $\mathcal{E} \subset \mathcal{N} \times \mathcal{N}$ is the set of edges or links. The edge $(i,j) \in \mathcal{E}$ denotes that node $i$ can obtain information from node $j$. Let $\mathcal{N}_i=\left\{ {j \in \mathcal{N}|(i,j) \in \mathcal{E}} \right\}$ denote the set of neighbours of node $i$. The graph $\mathcal{G}$ is said to be undirected if $(i,j) \in \mathcal{E}$ then $(j,i) \in \mathcal{E}$. $\mathcal{G}$ is said to be connected if there exists a sequence of edges between every pair of nodes. For an undirected graph $\mathcal{G}$, we may assign an orientation to $\mathcal{G}$ by considering one of the two nodes of a link to be the positive end and the other one to be the negative end. Denoting by $\mathscr{L}_i^ +$ (resp., $\mathscr{L}_i^ -$) the set of links for which node $i$ is the positive (resp. negative) end. Let $p$ be the cardinality $\mathcal{E}$, i.e., the total number of links. Define the incidence matrix $D=[d_{ik}]\in\mathbb{R}^{n\times p}$ of an undirected graph $\mathcal{G}$ as
\begin{equation*}
{d_{ik}} = \left\{ 
\begin{matrix}
     + 1, & k \in \mathscr{L}_i^ + \\
 - 1, & k \in \mathscr{L}_i^ - \\
0,& \mathrm{otherwise}.
\end{matrix}
\right.
\end{equation*}
For an undirected graph $\mathcal{G}$, it holds that $D^{\top} \textbf{1}_n =0 $. A spanning tree in $\mathcal{G}$ is an edge-subgraph of $\mathcal{G}$ which has $n-1$ edges and contains no circuits \cite[p29]{biggs1993algebraic}. A star graph is the graph that consists of one central node connected directly to multiple other nodes, which are not connected to each other.

\subsection{Passivity}
In this work, we adopt the definitions of passivity and input feedforward passivity in \cite{DesVid75} for system described by input-output maps.
\begin{definition}\label{def: passive}
A causal operator $H:\Ltwoe \to \Ltwoe $ is said to be passive if there exists some constant $\beta \in \mathbb{R}$ such that
\begin{equation}\label{eq: passive}
 {\left\langle {u,Hu } \right\rangle _T}\ge \beta,\,\forall u \in \Ltwoe ,\,\forall T \ge 0,
\end{equation}
and input-feedforward passive (IFP) if there exist $\nu \in \mathbb{R}$ and $\beta \in \mathbb{R}$ such that
\begin{align}\label{eq: IFP}
 {\left\langle {u,Hu } \right\rangle _T}\ge \nu\left\| {u} \right\|_T^2 + \beta,\forall u \in \Ltwoe,\,\forall T \ge 0.
\end{align}
\end{definition}
 
We say that $H$ is $\nu$-IFP if \eqref{eq: IFP} holds. The positive (or negative) sign of $\nu$ indicates the surplus (or shortage) of passivity of $H$ in \eqref{eq: IFP}. Intuitively, we are interested in the largest $\nu$. 
% A $\nu$-IFP system is said to be input strictly passive when $\nu > 0$. 

\section{Motivating Examples}\label{sec: motivating examples}
\begin{figure}[!ht]
\centering
\includegraphics[width=7cm]{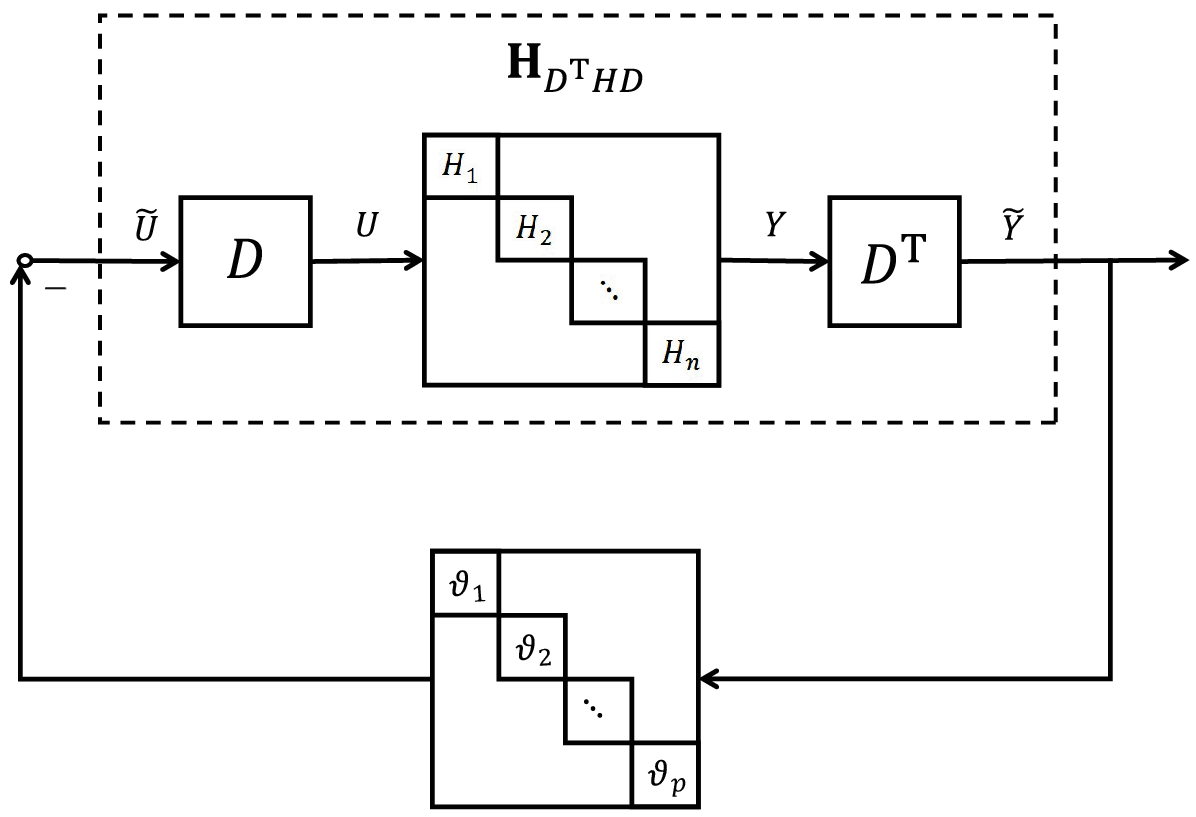}
\caption{Block diagram of the interconnected network: $H_i, \,i\in\{1,2,\ldots,n\}$ is the individual agent dynamics and $\vartheta_i, \,i\in\{1,2,\ldots,p\}$ denotes the coupling at each edge.}
\label{open_loop}
\end{figure}

A dynamical network with diffusive coupling over an undirected graph can be represented by the block diagram shown in Fig.~\ref{open_loop}.
Based on the feedback configuration in Fig.~\ref{open_loop}, the passivity theorem can be applied to establish output consensus in the network, provided that the open-loop system $\mathbf{H}_{D^{\top} H D}$ (enclosed by the dashed box) is passive, and that the coupling operators $\vartheta_i(\cdot)$, $i \in \{1, 2, \ldots, p\}$, are strictly passive. For the remainder of this work, we refer to the system inside the dashed box as the open-loop system $\mathbf{H}_{D^{\top} H D}$.
We observe that  ${U^{\top}}Y = {(D\tilde{U})^{\top}}Y = {\tilde{U}^{\top}}\tilde{Y}$, which means passivity from $\tilde{U}$ to $\tilde{Y}$ implies a weaker form of passivity from $U$ to $Y$\footnote{Note that the signal space for $U$ is constrained by the subspace $\mathrm{Image}(D)$.}. This provides the possibility of obtaining passivity of the open-loop system $\mathbf{H}_{D^{\top}HD}$ without requiring the dynamics of every individual agent $H_i$ to be passive. Based on this observation, and motivated by the analysis of two simple network examples to be presented next, we examine whether the open-loop system $\mathbf{H}_{D^{\top}HD}$ is passive in the presence of non-passive agents. \par

\subsection{A two-agent network}
\begin{figure}[!ht]
\centering
\includegraphics[width=5cm]{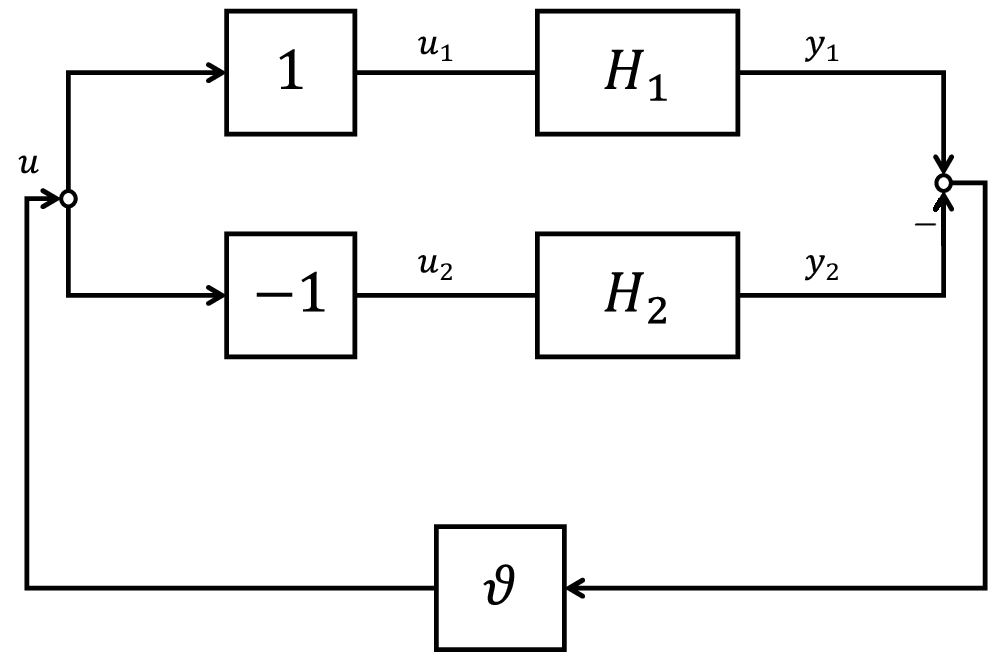}
\caption{Block diagram of the two-agent network.}
\label{two-agent network}
\end{figure}
Consider a network consisting of only two LTI agents: $H_1: u_1 \mapsto y_1$ and $H_2: u_2 \mapsto y_2$. Suppose that $H_1$ is $-\nu$-IFP and $H_2$ is $\nu$-IFP, where $\nu \geq 0$. Since there are only two agents, the incidence matrix is  $D = [1 \ -1]^\top$. The block diagram of this two-agent network is shown in Fig. \ref{two-agent network}. Noting that $y_1 = H_1 u_1 = H_1 u$ and $y_2 = H_2 u_2 = -H_2 u$, it follows that the open-loop system $\mathbf{H}_{D^{\top}HD}$ mapping $\tilde{U} = u$ to $\tilde{Y} = y_1 - y_2$ is given by $H_1 + H_2$. Since $H_1$ is $-\nu$-IFP and $H_2$ is $\nu$-IFP, it follows immediately that $H_1 + H_2$ is passive, which demonstrates that the shortage of passivity in $H_1$ can be compensated by the surplus of passivity in $H_2$, rendering the open-loop system $\mathbf{H}_{D^{\top}HD}$ passive.\par

\subsection{A three-agent network}
\begin{figure}[!ht]
\centering
\includegraphics[width=3cm]{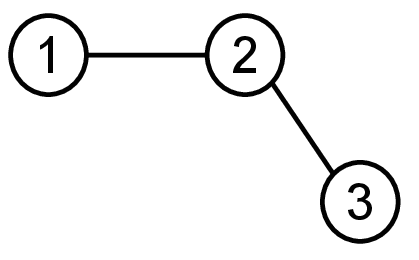}
\caption{Three-agent network.}
\label{three-agent network}
\end{figure}
Consider the network in Fig. \ref{three-agent network} with three LTI agents: $H_1: u_1\mapsto y_1$, $H_2: u_2\mapsto y_2$ and $H_3: u_3\mapsto y_3$. Suppose $H_1, H_3$ are IFP with $\nu \ge 0$ and $H_2$ is IFP with $\hat \nu \le 0$.  The incidence matrix of the graph (in Fig. \ref{three-agent network}) can be chosen to be $$D=\begin{bmatrix}
    1 & 0\\
    -1 & 1\\
    0  & -1
\end{bmatrix}.$$
Then it follows that the open-loop system $\mathbf{H}_{D^{\top}HD}$ maps $\tilde U=\TwoOne{\tilde{u}_1}{\tilde{u}_2}$ to $\tilde Y=\TwoOne{y_1-y_2}{y_2-y_3}$. Next, it can be derived that $u_1=\tilde{u}_1$, $u_2=-\tilde{u}_1+\tilde{u}_2$ and $u_3=-\tilde{u}_2$, and therefore $y_1-y_2=\left( H_1+H_2 \right)\tilde{u}_1-H_2\tilde{u}_2$ and $y_2-y_3= -H_2\tilde{u}_1+\left( H_2+H_3 \right)\tilde{u}_2$. Based on this, we can write the inner product of the input and output of the open-loop system $\mathbf{H}_{D^{\top}HD}$ as
\begin{align*}
 & {\left\langle { \TwoOne{\tilde{u}_1}{\tilde{u}_2}, \TwoOne{y_1-y_2}{y_2-y_3} } \right\rangle _T}\\
=& {\left\langle { \tilde{u}_1, y_1-y_2}\right\rangle _T} +{\left\langle { \tilde{u}_2, y_2-y_3}\right\rangle _T} \\
  % &= {\left\langle { \tilde{u}_1, \left( H_1+H_2 \right) \tilde{u}_1}\right\rangle _T} +{\left\langle { \tilde{u}_2,\left( H_2+H_3 \right)\tilde{u}_2 }\right\rangle _T}
  % - {\left\langle { \tilde{u}_1, H_2 \tilde{u}_2}\right\rangle _T} -{\left\langle { \tilde{u}_2, H_2\tilde{u}_1} \right\rangle _T}\\
=& {\left\langle { \tilde{u}_1, H_1 \tilde{u}_1}\right\rangle _T} +{\left\langle { \tilde{u}_2,H_3\tilde{u}_2} \right\rangle _T} + {\left\langle { \tilde{u}_1-\tilde{u}_2, H_2 (\tilde{u}_1-\tilde{u}_2)}\right\rangle _T}. 
\end{align*}
Since $H_1, H_3$ are IFP with $\nu \ge 0$ and $H_2$ is IFP with $\hat \nu \le 0$, one has
\begin{align*}
 & {\left\langle { \TwoOne{\tilde{u}_1}{\tilde{u}_2}, \TwoOne{y_1-y_2}{y_2-y_3} } \right\rangle _T}  \\
 \ge &\,  \nu \left( {\left\| {{{\tilde u}_1}} \right\|_T^2 + \left\| {{{\tilde u}_2}} \right\|_T^2} \right) + \hat \nu \left\| {{{\tilde u}_1} - {{\tilde u}_2}} \right\|_T^2+\bar \beta\\
 = & \,\left( {\nu  + \hat \nu } \right)\left( {\left\| {{{\tilde u}_1}} \right\|_T^2 + \left\| {{{\tilde u}_2}} \right\|_T^2} \right) - 2\hat \nu {\left\langle {{{\tilde u}_1},{{\tilde u}_2}} \right\rangle _T}+\bar \beta,
\end{align*}
where $\bar \beta = \sum_{i = 1}^3{{\beta _i}}$. Then it follows from $2\hat \nu {\left\langle {{{\tilde u}_1},{{\tilde u}_2}} \right\rangle _T} \le 2\left| {\hat \nu } \right|{\left\| {{{\tilde u}_1}} \right\|_T}{\left\| {{{\tilde u}_2}} \right\|_T}\le \left| {\hat \nu } \right|\left\| ({{{\tilde u}_1}} \right\|_T^2 + \left\| {{{\tilde u}_2}} \right\|_T^2)$ that ${\langle {\tilde U,\tilde Y} \rangle _T} \ge \bar \beta$ if $\nu  + \hat \nu  \ge \left| {\hat \nu } \right|$. That is, the open-loop system $\mathbf{H}_{D^{\top}HD}$ mapping $\tilde U$ to $\tilde Y$ is passive if $|\hat \nu| \le 0.5 \nu$.\par
These two simple examples demonstrate that the deficiency in passivity in one agent of a network can be compensated for by excess passivity in other agents.

\section{Problem Formulation}\label{sec: problem formulation}
\begin{figure}
\centering
\includegraphics[width=6cm]{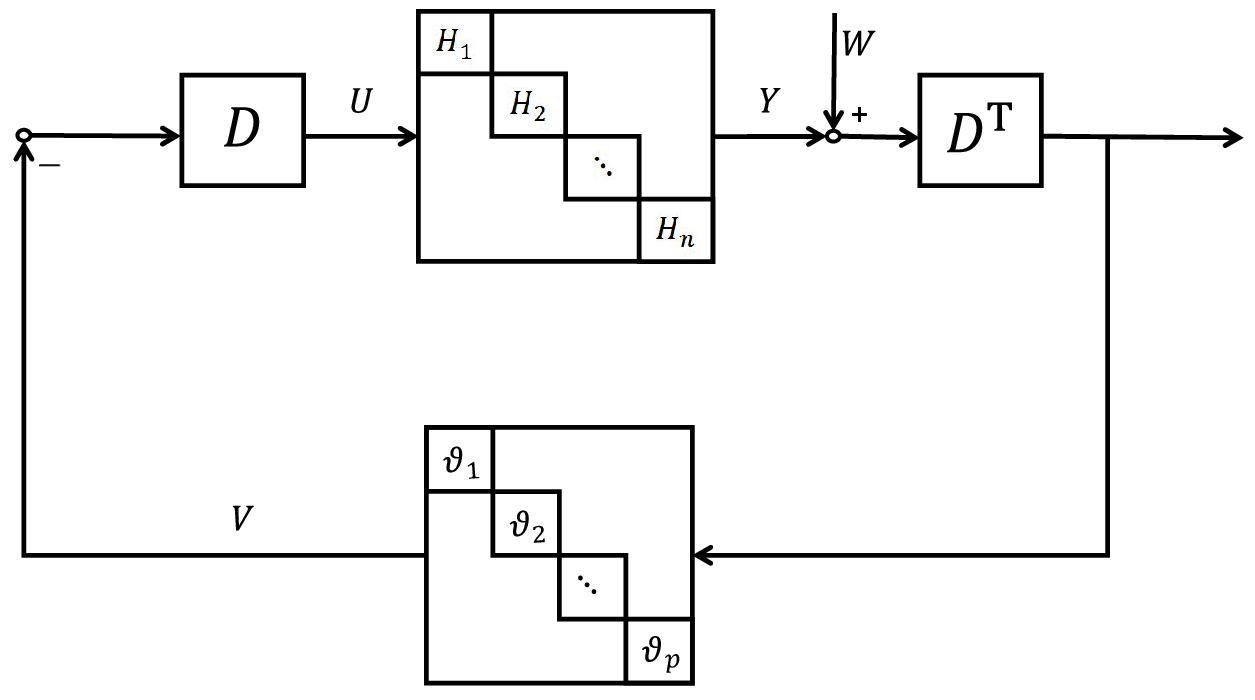}
\caption{Block diagram of the network in \eqref{eq: system model} \& \eqref{eq: input}.}
\label{fig.note2}
\end{figure}
Consider a group of $n$ systems $H_i:\Ltwoe \to\Ltwoe$ given by
\begin{equation}\label{eq: system model}
    {y_i} = {H_i}{u_i}, \,i\in\{1,2,\ldots,n\},
\end{equation}
where $u_i, y_i \in \Ltwoe$ denote respectively the input and output of the $i$-th system. 
% Note that to achieve non-trivial output consensus,  the transfer function $H_i$ are allowed to have non-repeated  poles on the imaginary axis. 
Suppose the group of systems is interconnected via an undirected and connected graph $\mathcal{G} = (\mathcal{N}, \mathcal{E})$. Specifically, the input $u_i$ to the $i$-th system is given by
\begin{equation}\label{eq: input} 
u_i = -\sum\limits_{j \in \mathcal{N}_i} {{\vartheta_{ij}}\left( {y_i + w_i - y_j - w_j} \right)}, \quad i \in \mathcal{N}. 
\end{equation}
Here, the external input $w_i \in \Ltwoe, i \in \{1,2,\ldots,n\}$ may represent measurement noise at the $i$-th system, or $(w_i - w_j) \in \Ltwoe$ may represent communication noise present on the links connecting the $i$-th and $j$-th systems. The operator ${\vartheta_{ij}}: \Ltwoe \to \Ltwoe$ mapping $0$ to $0$ is static and satisfies:
\begin{enumerate}
\item $\left( {{\vartheta _{ji}}\left( -x \right)} \right)\left( t \right)= - \left( {{\vartheta _{ij}}\left( x \right)} \right)\left( t \right)$  reflecting the undirected nature of the graph,
\item $0 < \underline{\alpha}_{ij}  \le \frac{\left( {{\vartheta _{ij}}\left( x \right)} \right)\left( t \right)}{{x\left( t \right)}} \le \overline{\alpha}_{ij} < \infty$ for all nonzero $x(t)$.
\end{enumerate}

Let $Y := {\rm col}\left( {{y_1}, \ldots ,{y_n}} \right)$ and the same notation is used to define the vectors $U$ and $W$. 
Recalling the definition of incidence matrix $D$, it can be obtained that
\begin{equation}\label{eq: iuput vector}
U = - D\Psi\left(D^{\top} \left( {Y + W} \right) \right), 
\end{equation}
where $\Psi : \Ltwoe \to \Ltwoe$ is component-wise defined as
\begin{equation}\label{eq:Psi}
\Psi \left( {\left[ {\begin{array}{*{20}{c}}
{{a_1}}\\
 \vdots \\
{{a_p}}
\end{array}} \right]} \right) = \left[ {\begin{array}{*{20}{c}}
{{b_1}}\\
 \vdots \\
{{b_p}}
\end{array}} \right]
% \Psi \left( {{{\left[ {\begin{array}{*{20}{c}}
% {{a_1}}& \cdots &{{a_p}}
% \end{array}} \right]}^{\top}}} \right) = {\left[ {\begin{array}{*{20}{c}}
% {{b_1}}& \cdots &{{b_p}}
% \end{array}} \right]^{\top}}
\end{equation}
such that ${b_k} = {\vartheta _k}\left( {{a_k}} \right),\,\forall k \in \{1,\dots, p\}$ with ${\vartheta _k}\left(  \cdot  \right) = {\vartheta _{ij}}\left(  \cdot  \right)$ if $d_{ik}=1$ and $d_{jk}=-1$.  Each component $\vartheta_k(\cdot)$ in $\Psi$ represents the coupling associated with the $k$-th edge. Define
\begin{align}\label{eq:V}
  V:=\Psi\left(D^{\top} \left( {Y + W} \right) \right).  
\end{align}
The network under consideration is represented by the block diagram in Fig.~\ref{fig.note2}. The only difference between the frameworks in Fig.~\ref{fig.note2} and Fig.~\ref{open_loop} is the inclusion of the external disturbance input $W$. They are equivalent when $W=0$.

\begin{definition}\label{def: IO consensus}
The network \eqref{eq: system model}, \eqref{eq: input} is said to achieve input-output consensus if there exist a finite gain $\rho>0$ and a constant $\sigma\ge 0$ such that
\begin{align}\label{eq:aim}
    {\left\| {D^{\top} Y} \right\|_T}\le \rho{\left\| {D^{\top} W} \right\|_T}+\sigma,\,\forall {W } \in \Ltwoe,\,\forall T \ge 0.
\end{align}
\end{definition}

In the absence of noise, i.e., when $W = 0$, the constant $\sigma$ in \eqref{eq:aim} accounts for a possible transient process resulting from differences in the agents' initial response, before they asymptotically converge to consensus.
As remarked in  \cite{scardovi2010synchronization}, ${\left\| {D^{\top} Y} \right\|_T}$ quantifies the synchrony of the outputs in the time interval $[0,T]$, and \eqref{eq:aim} implies that the interconnected network enjoy the property that external input with a high level of consensus produces output with the same property.

For a group of IFP agents \eqref{eq: system model}, our goal is to investigate whether a shortage of passivity (IFP with a negative index) in some agents within the group can be compensated by the passivity surplus (IFP with a positive index) of other agents, in the sense of ensuring passivity of the open-loop system $\mathbf{H}_{D^{\top}HD}$. Furthermore, we also explore how passivity surplus present in the coupling links can be utilised to locally compensate for this shortage, thereby contributing to the network consensus.

%%%%%%%%%%%%%%%%%%%%%%%%%%%%%%%%%%%%%%%%%%%%%%

\section{Main Results}\label{sec: main results}
\subsection{Passivity Compensation in Arbitrary Networks}\label{sec: Limitations of Generalisation to Arbitrary Networks}
Building on the analysis of the two-agent and three-agent network examples presented in Section \ref{sec: motivating examples}, we now investigate whether the observed passivity compensation can be generalised to arbitrary network topologies. In this subsection, we  establish a form of negative result, demonstrating that compensation within the group of agent dynamics is only possible when at most one agent lacks passivity. We then characterise how this shortage of passivity can be compensated by the passivity surplus of other agents.\par

Consider the open-loop system $\mathbf{H}_{D^{\top}HD}:\tilde{U}\mapsto \tilde{Y}$  in Fig. \ref{open_loop}. Assuming that system $H_i: u_i\mapsto y_i, \,i\in\{1,2,\ldots,n\}$ is IFP with index $\nu_i$, one has
\begin{align}\label{note2_1}
{\langle {\tilde U,\tilde Y} \rangle _T} &= {\langle {D \tilde U,Y} \rangle _T} = {\left\langle {U,Y} \right\rangle _T}\nonumber\\
& \ge {\nu _1}\left\| {{u_1}} \right\|_T^2+\beta_1 +  \cdots {\nu _n}\left\| {{u_n}} \right\|_T^2 + \beta_n \nonumber\\
& = {\left\langle {U,\Xi U} \right\rangle _T}+\bar \beta = {\langle {\tilde U,{D^{\top}}\Xi D \tilde U} \rangle _T}+\bar \beta,
\end{align}
where $\Xi  = \mathrm{diag}\left\{ {{\nu _1}, \dots, {\nu _n}} \right\}$ and $\bar \beta = \sum_{i = 1}^n {{\beta _i}}$. The open-loop system $\mathbf{H}_{D^{\top}HD}$ is passive if and only if ${D^{\top}}\Xi D \succcurlyeq 0$. By the definition of incidence matrix,
it can be obtained that  ${{D^{\top}}\Xi} =[\xi_{ki}]\in\mathbb{R}^{p \times n}$, where ${\xi_{ki}} = \left\{ 
\begin{matrix}
     + \nu_i, & k \in \mathscr{L}_i^ + \\
 - \nu_i, & k \in \mathscr{L}_i^ - \\
0,& \mathrm{otherwise},
\end{matrix}
\right.$
% \begin{align*}
% {\xi_{ki}} = \left\{ 
% \begin{matrix}
%      + \nu_i, & k \in \mathscr{L}_i^ + \\
%  - \nu_i, & k \in \mathscr{L}_i^ - \\
% 0,& \mathrm{otherwise}.
% \end{matrix}
% \right.
% \end{align*}
with $\mathscr{L}_i^ +$ (resp., $\mathscr{L}_i^ -$) are the set of links for which node $i$ is the positive (resp. negative) end. 
It follows that  ${{D^{\top}}\Xi D} :=[\theta_{kl}]\in\mathbb{R}^{p \times p}$ where
\begin{align}\label{eq:defDEDT}
{\theta_{kl}} = \left\{
\begin{matrix}
   \nu_i+\nu_j, &  k=l \in \mathscr{L}_i^ + \cap \mathscr{L}_j^ - \\
     \nu_i, & k \in \mathscr{L}_i^ + \cap \mathscr{L}_j^ -, l \in \mathscr{L}_i^ + \\
 - \nu_i, & k \in \mathscr{L}_i^ + \cap \mathscr{L}_j^ -, l \in \mathscr{L}_i^ - \\
    -  \nu_j, & k \in \mathscr{L}_i^ + \cap \mathscr{L}_j^ -, l \in \mathscr{L}_j^ + \\
  \nu_j, & k \in \mathscr{L}_i^ + \cap \mathscr{L}_j^ -, l \in \mathscr{L}_j^ - \\
0,& \mathrm{otherwise}.
\end{matrix}
\right.
\end{align}

\begin{theorem} \label{thm: atmostone}
The open-loop system $\mathbf{H}_{D^{\top}HD}$ in Fig. \ref{open_loop}  is non-passive if two or more agents in $H_i: u_i\mapsto y_i, \,i\in\{1,2,\ldots,n\}$ have a shortage of passivity with negative IFP indices.
\end{theorem}

\begin{proof}
We prove this theorem by showing that the matrix ${{D^{\top}}\Xi D}$ cannot be positive semi-definite if two or more agents have negative IFP indices. 

For the  undirected and connected graph $\mathcal{G}$, let $\mathcal{G}_{ST}$ be any spanning tree of $\mathcal{G}$ and denote by $D_{ST}\in\mathbb{R}^{n\times (n-1)}$ the incidence matrix of $\mathcal{G}_{ST}$. 
Denote $\mathcal{E}_{ST} \subset \mathcal{E}$ as the edge set of $\mathcal{G}_{ST}$. Then the incidence matrix of $\mathcal{G}$ can be set as $D = \OneTwo{{D_{ST}}}{{D_R}}$, where $D_R \in \mathbb{R}^{n\times \left( {p - (n - 1)} \right)}$ is the incidence matrix corresponding to the edges not in the spanning tree. Since $\text{rank}(D)=\text{rank}(D_{ST}) = n-1$, $D$ and $\OneTwo{D_{ST}}{0}\in\mathbb{R}^{n\times p}$ have the same column space, i.e., they are column equivalent. By \cite[Definition 1.13.17]{petersen2012linear}, there exists an invertible matrix $Q \in \mathbb{R}^{p \times p}$ such that $D = \left[ {\begin{array}{*{20}{c}}
{{D_{ST}}}&0
\end{array}} \right]Q$, which implies that ${D^{\top}}\Xi D = Q^{\top}\TwoOne{D_{ST}^{\top}}{0}\Xi \OneTwo{D_{ST}}{0}{Q}$, i.e., ${D^{\top}}\Xi D$ and $\TwoOne{D_{ST}^{\top}}{0}\Xi\OneTwo{D_{ST}}{0}$ are congruent. Since $\TwoOne{D_{ST}^{\top}}{0}\Xi \OneTwo{D_{ST}}{0}=\TwoTwo{D_{ST}^{\top}\Xi {D_{ST}}}{0}{0}{0}$, one has ${D^{\top}}\Xi D \succcurlyeq 0$ is equivalent to $D_{ST}^{\top}\Xi {D_{ST}} \succcurlyeq 0$. 

Next, let $\mathcal{G}_{ST}^*$ denote any arbitrary star graph with $n$ nodes, and let $D_{ST}^*\in\mathbb{R}^{n\times (n-1)}$ be the corresponding incidence matrix. Suppose that there is a ``virtual" undirected and connected graph $\mathcal{G}_{v}$ that includes both $\mathcal{G}_{ST}^*$ and $\mathcal{G}_{ST}$ as spanning trees, and denote by $D_v$ its incidence matrix. By the previous argument, there exist invertible matrices $\bar Q$ and $\hat Q$ such that $D_v = \OneTwo{D_{ST}}{0}\bar Q = \OneTwo{D_{ST}^*}{0}\hat Q$. Hence, $\TwoOne{D_{ST}^{\top}}{0}\Xi \OneTwo{D_{ST}}{0}$ and $\TwoOne{D_{ST}^{*\top}}{0}\Xi \OneTwo{D_{ST}^*}{0}$ are congruent.
% It follows from  the fact that $D_{ST}$ can be transformed into $D_{ST}^*$ by a sequence of elementary column operations, $D_{ST}$ and $D_{ST}^*$ are column equivalent. 
Therefore, $D_{ST}^{\top}\Xi {D_{ST}} \succcurlyeq 0$ is equivalent to ${D_{ST}^{*{\top}}} \Xi {D_{ST}^*} \succcurlyeq 0$.
Suppose the $i$-th agent $H_i$ is IFP with negative index $\nu_i$, and let it be the centre of the  star topology. Recalling the derivation in \eqref{eq:defDEDT}, the diagonal elements of ${D_{ST}^{*{\top}}} \Xi {D_{ST}^*}$ are $\nu_i + \nu_j, j \in \{1,\ldots,n \}\backslash \{i\}$. Since $\nu_i < 0$, if any $\nu_j$ is negative, then at least one of the diagonal elements of ${D_{ST}^{*{\top}}} \Xi {D_{ST}^*}$ is negative, i.e., ${D_{ST}^{*{\top}}} \Xi {D_{ST}^*}$ is not positive semi-definite. Therefore, if two or more agents in $H_i: u_i\mapsto y_i, \,i\in\{1,2,\ldots,n\}$ have a shortage of IFP, ${{D^{\top}}\Xi D}$ cannot be positive semi-definite, which implies that the open-loop system $\mathbf{H}_{D^{\top}HD}$ in Fig. \ref{open_loop} is non-passive.
\end{proof}

\begin{remark}
Theorem~\ref{thm: atmostone} concerns the passivity of the open-loop system $\mathbf{H}_{D^\top H D} = D^\top \operatorname{diag}\{H_1, \ldots, H_n\} D$, which involves only the agent dynamics and the graph topology. It does not account for the nonlinear couplings at the links, represented by $\operatorname{diag}\{\vartheta_1, \ldots, \vartheta_p\}$, which constitute the other component of the feedback interconnection in Fig.~\ref{open_loop}.
\end{remark}

According to Theorem \ref{thm: atmostone}, the open-loop system $\mathbf{H}_{D^{\top}HD}$ can be passive only if at most one agent in the network is non-passive.  In the case where a single agent exhibits a negative IFP index, we propose in the following a condition in terms of passivity indices to ensure passivity of the open-loop system $\mathbf{H}_{D^{\top}HD}$. To this end, we first introduce the following lemma.

\begin{lemma}\label{lem: positive definite}
Given a matrix $A=A^\top = [a_{ij}]\in\mathbb{R}^{p\times p}$, it is positive semi-definite if there exists a diagonal  matrix $S=\mathrm{diag}\left\{ {{s _1}, \ldots, s_p} \right\}\succ 0$ such that ${a_{ii}}{s_i} \ge \sum_{j = 1,j \ne i}^p {\left| {{a_{ij}}} \right|{s_j}} $ for all $i = \{1, \ldots, p\}$. If these inequalities are strict, $A$ is positive definite.
\end{lemma}

\begin{proof}
First, observe that $SAS = [\bar a_{ij}] \in\mathbb{R}^{p\times p}$, where $\bar a_{ij} = a_{ij}s_is_j$. By the Gershgorin theorem \cite[p344]{horn2012matrix}, all eigenvalues of $SAS$ are located in the union of $p$ discs centred at $\bar a_{ii}$ of radius $\sum_{j = 1,j \ne i}^p {\left| {{{\bar a}_{ij}}} \right|}$ for $i = \{1, \ldots, p\}$. Therefore, if ${a_{ii}}s_i^2 \ge \sum_{j = 1,j \ne i}^p {\left| {{a_{ij}}} \right|{s_i}{s_j}}$, all eigenvalues of $SAS$ are non-negative, which implies that $SAS \succcurlyeq 0$. Since $S$ is an invertible matrix, $SAS \succcurlyeq 0$ is equivalent to $A\succcurlyeq 0$. Thus, if ${a_{ii}}{s_i} > \sum_{j = 1,j \ne i}^p {\left| {{a_{ij}}} \right|{s_j}} $,  $SAS \succ 0$, which implies that $A \succ 0$.
\end{proof}

\begin{proposition}\label{prop: nu_n}
Suppose that $H_i,i\in\{1,\ldots,n - 1\}$ are IFP with index $\nu_i>0$ and $H_n$ is IFP with index $\nu_n<0$. The open-loop system $\mathbf{H}_{D^{\top}HD}$ in Fig. \ref{open_loop} is passive  if there exist $s_i>0,i\in\{1,\ldots,n-1\}$ such that $${\nu _i} \ge \frac{{{s_1} + \dots + {s_{n - 1}}}}{{{s_i}}}\left| {{\nu _n}} \right|,i\in\{1,\ldots,n - 1\}.$$    
\end{proposition}

\begin{proof}
Consider a ``virtual" star graph $\mathcal{G}_{ST}^*$ connecting the $n$ agents $H_i, i\in\{1,\ldots,n\}$ with $H_n$ being the centre node. Let $D_{ST}$ be its incidence matrix. From the derivation in \eqref{eq:defDEDT}, it follows that ${D_{ST}^{*{\top}}} \Xi {D_{ST}^*}= [\theta_{ij}]\in\mathbb{R}^{(n-1) \times (n-1)}$, where ${\theta _{ii}} =  {\nu _i} + {\nu _n}$ and ${\theta _{ij}} = {\nu _n}, j \neq i$ for all $i\in\{1,\ldots,n-1\}$. According to Lemma  \ref{lem: positive definite}, ${D_{ST}^{*{\top}}} \Xi {D_{ST}^*}$ is positive semi-definite if there exists a diagonal  matrix $S=\mathrm{diag}\left\{ {{s _1}, \ldots, s_{n-1}} \right\}\succ 0$ such that ${\theta _{ii}}{s_i} \ge \sum_{j = 1,j \ne i}^{n - 1} {\left| {{\theta _{ij}}} \right|{s_j}} $, i.e., $\left( {{\nu _i} + {\nu _n}} \right){s_i} \ge \sum_{j = 1,j \ne i}^{n - 1} {\left| {{\nu _n}} \right|{s_j}}$, for all $i\in\{1,\ldots,n-1\}$. Since $\nu_n <0$, $\left( {{\nu _i} + {\nu _n}} \right){s_i} \ge \sum_{j = 1,j \ne i}^{n - 1} {\left| {{\nu _n}} \right|{s_j}}$ is equivalent to ${\nu _i}{s_i} \ge \left| {{\nu _n}} \right|\left( {{s_1} +  \cdots  + {s_{n - 1}}} \right)$. By the proof of Theorem \ref{thm: atmostone}, $D^{\top}\Xi {D} \succcurlyeq 0$ is equivalent to ${D_{ST}^{*{\top}}} \Xi {D_{ST}^*} \succcurlyeq 0$. Therefore, if there exist $s_i>0,i\in\{1,\ldots,n-1\}$ such that ${\nu _i} \ge \frac{{{s_1}+ \dots + {s_{n - 1}}}}{{{s_i}}}\left| {{\nu _n}} \right|,i\in\{1,\ldots,n - 1\}$, ${{D^{\top}}\Xi D}$ is positive semi-definite, which implies that the open-loop system $\mathbf{H}_{D^{\top}HD}$ is passive.
\end{proof}

\subsection{Passivity Compensation in Feedback Connection}\label{sec: Passivity Compensation in Feedback Connection}
In the previous subsection, we looked into the passivity compensation among agents within the open-loop system $\mathbf{H}_{D^{\top}HD}$.
It has been established that in cases where more than one agent in the network lacks passivity, relying solely on compensation within the group of agents is insufficient. In this subsection, we turn our attention to the passivity surplus present in the coupling links and explore how it can be utilised locally to compensate for the lack of passivity in the agent dynamics. Specifically, we present in the following theorem  a distributed condition for consensus, formulated in terms of passivity indices and coupling gains.

\begin{theorem}\label{thm: consensus condition}
Consider the network described by \eqref{eq: system model} and \eqref{eq: input}, where each agent $H_i$, $i \in \{1, \ldots, n\}$, is IFP with index $\nu_i$. Let $r_i$ denote the number of neighbours of agent $H_i$. The network achieves input-output consensus if the following condition holds for all edges $(i, j) \in \mathcal{E}$:
\[
    \frac{1}{\overline{\alpha}_{ij}} + \nu_i + \nu_j - (r_i - 1)|\nu_i| - (r_j - 1)|\nu_j| > 0,
\]
where $\overline{\alpha}_{ij}$ is the upper sector bound of the operator $\vartheta_{ij}(\cdot)$, representing the nonlinear coupling between agents $i$ and $j$.
\end{theorem}

\begin{proof}
% The proof is omitted due to space constraints and is available in cite (). 
Since  $0 < \underline{\alpha}_{ij}  \le \frac{\left( {{\vartheta _{ij}}\left( x \right)} \right)\left( t \right)}{{x\left( t \right)}} \le \overline{\alpha}_{ij} < \infty$ for all nonzero $x(t)$,  which implies that $\frac{1}{{{\overline{\alpha}_{ij}}}}\left\| {{\vartheta _{ij}}\left( x \right)} \right\|_T^2 \le {\left\langle {x,{\vartheta _{ij}}\left( x \right)} \right\rangle _T}$. It then follows from \eqref{eq:Psi} and \eqref{eq:V} that
\begin{align}\label{eq: V,Y}
& \, {\left\langle {V,{D^{\top}}\left( {Y + W} \right)} \right\rangle _T} \nonumber\\
=& \, \int_0^T {{{\left( {Y + W} \right)}^{\top}}D} \Psi \left( {{D^{\top}}\left( {Y + W} \right)} \right)dt \nonumber\\
\ge& \, \int_0^T {\Psi {{\left( {{D^ {\top} }\left( {Y + W} \right)} \right)}^\top}\Lambda \Psi \left( {{D^ {\top} }\left( {Y + W} \right)} \right)dt} \nonumber\\
=& \, {\left\langle {V,\Lambda V} \right\rangle _T},
\end{align}
where $\Lambda :=\text{diag}\{ {\alpha _1}, \ldots ,{\alpha _p}\}$ with ${\alpha _k} = \frac{1}{{{{\overline{\alpha} }_{ij}}}}$ if $d_{ik}=1$ and $d_{jk}=-1$.  Since the agents are IFP with index $\nu_i,i\in\{1,\ldots,n\}$, we have that
\begin{align}\label{eq: U,Y}
{\left\langle {U,Y} \right\rangle _T}& \ge {\nu _1}\left\| {{u_1}} \right\|_T^2 + \beta_1 +  \cdots {\nu _n}\left\| {{u_n}} \right\|_T^2 +\beta_n\nonumber\\
&= {\left\langle {U,\Xi U} \right\rangle _T}+\bar \beta,
\end{align}
where $\Xi  :=\mathrm{diag}\left\{ {{\nu _1}, \dots, {\nu _n}} \right\}$ and $\bar \beta  = \sum_{i = 1}^n {{\beta _i}} $. 
It can be derived that
\begin{align}
 {\left\langle {V,{D^{\top}}W} \right\rangle _T}
=& \,{\left\langle {V, - {D^{\top}}Y} \right\rangle _T} + {\left\langle {V,{D^{\top}}\left( {Y + W} \right)} \right\rangle _T}\nonumber\\
\overset{\mathrm{(a)}}{=}  & \,{\left\langle {U, Y} \right\rangle _T} + {\left\langle {V,{D^{\top}}\left( {Y + W} \right)} \right\rangle _T}\nonumber\\
\overset{\mathrm{(b)}}{\ge} & \, {\left\langle {U,\Xi U} \right\rangle _T}+{\left\langle {V,\Lambda V} \right\rangle _T}+\bar \beta \nonumber\\
\overset{\mathrm{(c)}}= & {\left\langle {V,\left( {{D^{\top}}\Xi D + \Lambda } \right)V} \right\rangle _T}+\bar \beta,\label{eq: V,W}
\end{align}
where (a) holds due to  $U=-DV$, (b) follows from \eqref{eq: U,Y}, (c) is obtained by combining \eqref{eq: V,Y} and the fact that ${\left\langle {U,\Xi U} \right\rangle _T} = {\left\langle {V, D^{\top}\Xi D V} \right\rangle _T}$. 
% \begin{align}\label{eq: V,W}
% & \, {\left\langle {V,\left( {{D^{\top}}\Xi D + \Lambda } \right)V} \right\rangle _T}+\bar \beta\nonumber\\
% = & \, {\left\langle {U,\Xi U} \right\rangle _T}+{\left\langle {V,\Lambda V} \right\rangle _T}+\bar \beta \nonumber\\
% \le & \,{\left\langle {U, Y} \right\rangle _T} + {\left\langle {V,{D^{\top}}\left( {Y + W} \right)} \right\rangle _T}\nonumber\\
% =& \,{\left\langle {V, - {D^{\top}}Y} \right\rangle _T} + {\left\langle {V,{D^{\top}}\left( {Y + W} \right)} \right\rangle _T}\nonumber\\
% =& \, {\left\langle {V,{D^{\top}}W} \right\rangle _T},
% \end{align}

From the derivation in \eqref{eq:defDEDT}, we have that  ${D^{\top}}\Xi D + \Lambda = [\zeta_{kl}]\in\mathbb{R}^{p \times p}$, where ${\zeta _{kk}} = \frac{1}{{{{\overline{\alpha} }_{ij}}}} + {\nu _i} + {\nu _j},\sum_{l = 1,l \ne k}^p {\left| {{\zeta _{kl}}} \right|}  = \left( {{r_i} - 1} \right)\left| {{\nu _i}} \right| + \left( {{r_j} - 1} \right)\left| {{\nu _j}} \right|$ if $d_{ik}=1$ and $d_{jk}=-1$. By hypothesis, $\frac{1}{{{{\overline{\alpha} }_{ij}}}} + {\nu _i} + {\nu _j} - \left( {{r_i} - 1} \right)\left| {{\nu _i}} \right| - \left( {{r_j} - 1} \right)\left| {{\nu _j}} \right| > 0$ for all $(i,j)\in\mathcal{E}$. We obtain from Lemma \ref{lem: positive definite} that ${D^{\top}}\Xi D + \Lambda$ is positive definite. Let $\kappa >0$ be the smallest eigenvalue of ${D^{\top}}\Xi D + \Lambda$, and then from \eqref{eq: V,W} we obtain
\begin{align*}
\kappa \left\| V \right\|_T^2 &\le {\left\langle {V,\left( {{D^{\top}}\Xi D + \Lambda } \right)V} \right\rangle _T} \le {\left\langle {V,{D^{\top}}W} \right\rangle _T} - \bar \beta\\   
& \le {\left\langle {V,{D^{\top}}W} \right\rangle _T} - \bar \beta + \frac{1}{2}\left\| {\sqrt \kappa V - \frac{1}{{\sqrt \kappa }}{D^{\top}}W} \right\|_T^2\\
& = \frac{\kappa }{2}\left\| V \right\|_T^2 + \frac{1}{{2\kappa}}\left\| {{D^{\top}}W} \right\|_T^2 - \bar \beta,
\end{align*}
which leads to
$\left\| V \right\|_T^2 \le \frac{1}{{{\kappa^2}}}\left\| {{D^{\top}}W} \right\|_T^2 - \frac{{2\bar \beta }}{\kappa }.$ 
By the fact that ${a^2} \pm {b^2} \le  {\left( {\left| a \right| + \left| b \right|} \right)^2}$, this yields
\begin{align}\label{eq: norm V}
{\left\| V \right\|_T} \le \frac{1}{\kappa}{\left\| {{D^{\top}}W} \right\|_T} + \sqrt {\frac{{2\left| {\bar \beta } \right|}}{\kappa }}.
\end{align}
Noting from the definition of  $\Psi$ in \eqref{eq:Psi} and \eqref{eq:V} that
\begin{align}\label{eq: norm V 2}
{\left\| {V} \right\|_T} \ge \underline{\alpha} {\left\| {{D^{\top}}\left( {Y + W} \right)} \right\|_T},
\end{align}
where $\underline{\alpha} := \mathop {\min }\limits_{(i,j) \in \mathcal{E}} {\underline{\alpha}_{ij}}$ with ${\underline{\alpha}_{ij}}$ is the lower sector bound of $\vartheta_{ij}(\cdot)$. It follows from \eqref{eq: norm V}, \eqref{eq: norm V 2} and $\left| {a + b} \right| \ge \left| a \right| - \left| b \right|$ that
$${\left\| {{D^{\top}}Y} \right\|_T} \le \left( {\frac{1}{{\kappa \underline{\alpha} }} + 1} \right){\left\| {{D^{\top}}W} \right\|_T} + \frac{1}{{\underline{\alpha} }} \sqrt {\frac{{2\left| {\bar \beta } \right|}}{\kappa }}, \, \forall T\ge 0,$$
which completes the proof.
\end{proof}

\begin{example}
\begin{figure}[!ht]
\centering
\includegraphics[width=4.5cm]{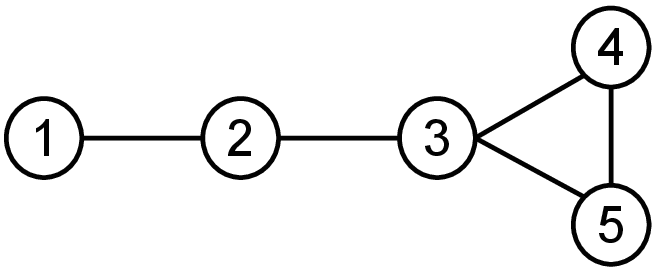}
\caption{The network considered in Example 1.}
\label{Fig4}
\end{figure}
Consider the network in Fig. \ref{Fig4}, where the dynamics of the agents are given by 
\begin{align*}
&{H_1} = \frac{{s + 0.8}}{{s(s+0.57)}}, {H_2} = \frac{{s + 1}}{{s(s+0.7)}}, {H_3} = \frac{{s + 1.5}}{{s(s+1)}},\\
&{H_4} = \frac{{\left( {s + 0.45} \right)\left( {s + 0.65} \right)}}{{s\left( {s + 0.4} \right)\left( {s + 0.6} \right)}}, {H_5} =\frac{{\left( {s + 0.5} \right)\left( {s + 0.9} \right)}}{{s\left( {s + 0.43} \right)\left( {s + 0.8} \right)}}.
\end{align*}
These agent dynamics $H_i$ exhibit significantly different set of zeros and poles. The incidence matrix of the graph (in Fig. \ref{Fig4}) is given by
$$D=\begin{bmatrix}
     1  &   0   &  0  &   0  &   0\\
    -1  &   1   &  0  &   0  &   0\\
     0  &  -1   &  1  &   1  &   0\\
     0  &   0   & -1  &   0  &   1\\
     0  &   0   &  0  &  -1  &  -1
\end{bmatrix}.$$
Suppose the agents are coupled by the following sector-bounded functions:
$${\vartheta _{ij}(x)}=\begin{cases}
{a_{ij}\sin( x )},&\text{if } \left| {x} \right| < \frac{\pi }{2},\\
a_{ij} x,&\mathrm{otherwise},
\end{cases}\quad (i,j) \in \mathcal{E}$$
with $a_{12}= 0.65$, $a_{23}= 0.40$, $a_{34}= 0.34$, $a_{35} = 0.33$ and $a_{45} = 0.44$. It can be verified that ${{\overline{\alpha} }_{ij}} = a_{ij}$ for all $(i,j) \in \mathcal{E}$. Next, by solving the LMI  in \cite[Lemma 2]{kottenstette2014relationships}, we obtain the IFP index $\nu_i$ for each agent: $\nu_1 = -0.71$, $\nu_2 = -0.41$, $\nu_3 = -0.55$, $\nu_4=-0.50$, and $\nu_5 = -0.61$. Note that the number of neighbours of each agent are $r_1 =1$, $r_2 =2$, $r_3 =3$, $r_4 = 2$ and $r_5 =2$. Therefore, it can be verified that $\frac{1}{{{\overline{\alpha} }_{ij}}} + {\nu _i} + {\nu _j} - \left( {{r_i} - 1} \right)\left| {{\nu _i}} \right| - \left( {{r_j} - 1} \right)\left| {{\nu _j}} \right| > 0$ for all $(i,j) \in \mathcal{E}$. By Theorem \ref{thm: consensus condition}, 
% $${D^{\top}}\Xi D + \Lambda = \begin{bmatrix}
%    0.4185  &  0.4100    &     0      &   0   &      0\\
%    0.4100  &  1.5400  &  0.5500  &  0.5500   &      0\\
%         0  &  0.5500  &  1.7712  & -0.5500   & 0.6200\\
%         0  &  0.5500  & -0.5500  &  1.6527   &-0.5000\\
%         0   &      0  &  0.6200 &  -0.5000    &1.1527
% \end{bmatrix} \succ 0,$$
% and 
the  network will achieve IO consensus. With initial value $Y(0) ={\left[ {\begin{array}{*{20}{c}}
-0.3 & -0.25 & -0.625 & 0.5963 & -0.2725
\end{array}} \right]^{\top}} $ and the measurement noise and/or communication noise $w_i(t) = 0.1{\bar w_i}(t)$, where ${\bar w_i}(t)$ is white Gaussian noise with ${\bar w_i}(t) \sim \mathcal{N}(0,1)$, the outputs of the agents reach consensus approximately with ${\left\| {D^{\top} Y} \right\|_T}$ bounded in terms of ${\left\| {D^{\top} W} \right\|_T}$, as shown in  Fig. \ref{example1}.
\begin{figure}[!ht]
\centering
\includegraphics[width=6.5cm]{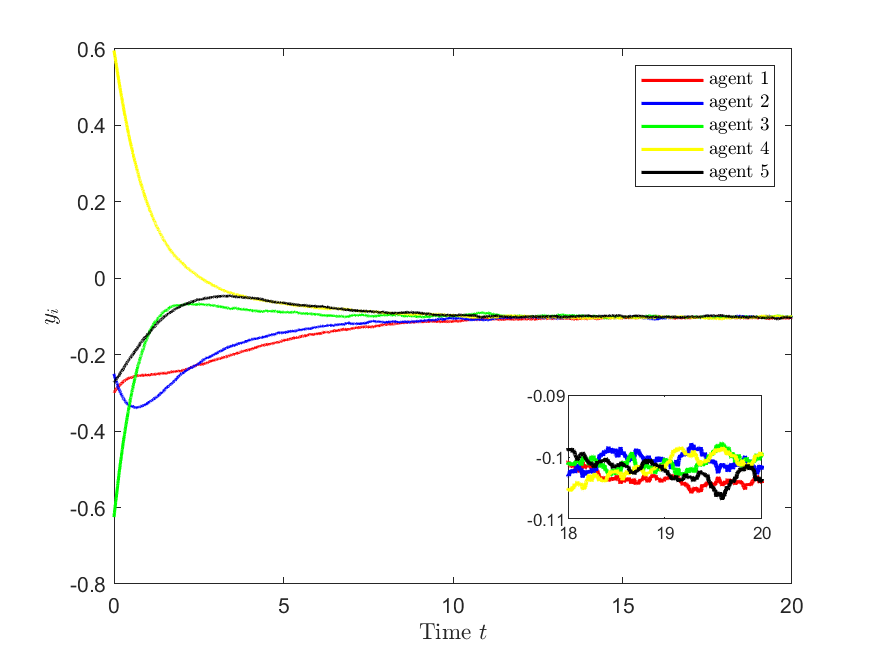}
\caption{Output trajectories of the agents in Example 1.}
\label{example1}
\end{figure}
\end{example}

\section{Conclusion}\label{sec: Conclusion}
We proposed a passivity compensation-based approach for output consensus in heterogeneous networks with nonlinear couplings, under measurement and communication noise. Focusing on input-feedforward passive (IFP) agents, we showed that agent-to-agent passivity compensation is only feasible when at most one agent lacks passivity, and provided a scheme to offset this deficiency using the surplus of other agents.  We also addressed compensation between the lack of passivity in agents and the surplus in the coupling links, deriving a distributed condition based on passivity indices and coupling gains to guarantee output consensus.

% Future work will explore a plug-and-play framework for dynamic network expansion and broader classes of nonlinear coupling.

% Finally, building on this distributed condition, we propose a plug-and-play framework that allows new agents to be seamlessly added to the network, with only local conditions required to be checked, while preserving robust consensus. The practical relevance of this framework is illustrated through a case study.
\bibliographystyle{IEEEtran}
\bibliography{reference}

% Generated by IEEEtran.bst, version: 1.14 (2015/08/26)
\begin{thebibliography}{10}
\providecommand{\url}[1]{#1}
\csname url@samestyle\endcsname
\providecommand{\newblock}{\relax}
\providecommand{\bibinfo}[2]{#2}
\providecommand{\BIBentrySTDinterwordspacing}{\spaceskip=0pt\relax}
\providecommand{\BIBentryALTinterwordstretchfactor}{4}
\providecommand{\BIBentryALTinterwordspacing}{\spaceskip=\fontdimen2\font plus
\BIBentryALTinterwordstretchfactor\fontdimen3\font minus \fontdimen4\font\relax}
\providecommand{\BIBforeignlanguage}[2]{{%
\expandafter\ifx\csname l@#1\endcsname\relax
\typeout{** WARNING: IEEEtran.bst: No hyphenation pattern has been}%
\typeout{** loaded for the language `#1'. Using the pattern for}%
\typeout{** the default language instead.}%
\else
\language=\csname l@#1\endcsname
\fi
#2}}
\providecommand{\BIBdecl}{\relax}
\BIBdecl

\bibitem{willems1972dissipative1}
J.~C. Willems, ``Dissipative dynamical systems part {I}: General theory,'' \emph{Archive for rational mechanics and analysis}, vol.~45, no.~5, pp. 321--351, 1972.

\bibitem{DesVid75}
C.~A. Desoer and M.~Vidyasagar, \emph{Feedback Systems: Input-Output Properties}.\hskip 1em plus 0.5em minus 0.4em\relax Academic Press, 1975.

\bibitem{van2000l2}
A.~Van~der Schaft, \emph{$L_2$-gain and passivity techniques in nonlinear control}.\hskip 1em plus 0.5em minus 0.4em\relax Springer, 2000.

\bibitem{arcak2007passivity}
M.~Arcak, ``Passivity as a design tool for group coordination,'' \emph{IEEE Transactions on Automatic Control}, vol.~52, no.~8, pp. 1380--1390, 2007.

\bibitem{scardovi2010synchronization}
L.~Scardovi, M.~Arcak, and E.~D. Sontag, ``Synchronization of interconnected systems with applications to biochemical networks: An input-output approach,'' \emph{IEEE Transactions on Automatic Control}, vol.~55, no.~6, pp. 1367--1379, 2010.

\bibitem{hamadeh2011global}
A.~Hamadeh, G.-B. Stan, R.~Sepulchre, and J.~Gon{\c{c}}alves, ``Global state synchronization in networks of cyclic feedback systems,'' \emph{IEEE Transactions on Automatic Control}, vol.~57, no.~2, pp. 478--483, 2011.

\bibitem{chopra2006passivity}
N.~Chopra and M.~W. Spong, ``Passivity-based control of multi-agent systems,'' \emph{Advances in robot control: from everyday physics to human-like movements}, pp. 107--134, 2006.

\bibitem{burger2014duality}
M.~B{\"u}rger, D.~Zelazo, and F.~Allg{\"o}wer, ``Duality and network theory in passivity-based cooperative control,'' \emph{Automatica}, vol.~50, no.~8, pp. 2051--2061, 2014.

\bibitem{jain2018regularization}
A.~Jain, M.~Sharf, and D.~Zelazo, ``Regularization and feedback passivation in cooperative control of passivity-short systems: A network optimization perspective,'' \emph{IEEE Control Systems Letters}, vol.~2, no.~4, pp. 731--736, 2018.

\bibitem{sharf2020geometric}
M.~Sharf, A.~Jain, and D.~Zelazo, ``Geometric method for passivation and cooperative control of equilibrium-independent passive-short systems,'' \emph{IEEE Transactions on Automatic control}, vol.~66, no.~12, pp. 5877--5892, 2020.

\bibitem{bai2011cooperative}
H.~Bai, M.~Arcak, and J.~Wen, \emph{Cooperative control design: a systematic, passivity-based approach}.\hskip 1em plus 0.5em minus 0.4em\relax Springer Science \& Business Media, 2011.

\bibitem{li2019consensus}
M.~Li, L.~Su, and G.~Chesi, ``Consensus of heterogeneous multi-agent systems with diffusive couplings via passivity indices,'' \emph{IEEE Control Systems Letters}, vol.~3, no.~2, pp. 434--439, 2019.

\bibitem{sharf2019network}
M.~Sharf and D.~Zelazo, ``Network feedback passivation of passivity-short multi-agent systems,'' \emph{arXiv preprint arXiv:1902.08986}, 2019.

\bibitem{biggs1993algebraic}
N.~Biggs, \emph{Algebraic graph theory}.\hskip 1em plus 0.5em minus 0.4em\relax Cambridge university press, 1974.

\bibitem{petersen2012linear}
P.~Petersen, \emph{Linear algebra}.\hskip 1em plus 0.5em minus 0.4em\relax Springer, 2012.

\bibitem{horn2012matrix}
R.~A. Horn and C.~R. Johnson, \emph{Matrix analysis}.\hskip 1em plus 0.5em minus 0.4em\relax Cambridge university press, 1990.

\bibitem{kottenstette2014relationships}
N.~Kottenstette, M.~J. McCourt, M.~Xia, V.~Gupta, and P.~J. Antsaklis, ``On relationships among passivity, positive realness, and dissipativity in linear systems,'' \emph{Automatica}, vol.~50, no.~4, pp. 1003--1016, 2014.

\end{thebibliography}
\end{document}